%% file: flops2018.tex
\documentclass{llncs}
\pdfoutput=1 

\usepackage[utf8]{inputenc}

\usepackage{booktabs} 

\usepackage{amsmath}
\usepackage{amssymb}

\usepackage{microtype}

\usepackage{enumitem}

\usepackage{tikz}

\usetikzlibrary{shapes}
\usetikzlibrary{positioning,calc}

\usepackage{hyperref}
\usepackage{cleveref}

\input{confmodeq-commands}

\begin{document}

\title{Confluence Modulo Equivalence with Invariants in Constraint Handling Rules}
\titlerunning{Confluence Modulo Equivalence with Invariants in CHR}  
%
\author{Daniel~Gall \and
Thom~Frühwirth}
\authorrunning{D. Gall and T. Frühwirth} 
\institute{Institute of Software Engineering and Programming Languages, Ulm University, \\89069 Ulm, Germany\\
\email{daniel.gall@uni-ulm.de}, \email{thom.fruehwirth@uni-ulm.de}}

\maketitle              

\pagestyle{headings}  

\input{confmodeq-abstract}


\input{confmodeq-body}

%
%
\bibliographystyle{splncs}
\bibliography{bib} 

\appendix
\input{confmodeq-appendix}

\end{document}

%% file: confmodeq-commands.tex
\newcommand*{\constr}[1]{\mathord{\mathit{#1}}}

\newcommand{\inv}{\mathcal{I}}

\newcommand{\eq}{\approx}
\newcommand{\seteq}{\approx^\mathrm{S}}
\newcommand{\nseteq}{\noteq^\mathrm{S}}
\newcommand{\noteq}{\not\approx}
\newcommand{\merge}{\diamond}
\newcommand{\ltstate}{\vartriangleleft}
\newcommand{\joins}{\downarrow}

\newlist{invariant}{enumerate}{1}
\setlist[invariant]{label=(B\arabic*)}
\crefname{invarianti}{invariant}{invariants} 
\Crefname{invarianti}{Invariant}{Invariants}

%% file: confmodeq-abstract.tex
\begin{abstract}
Confluence denotes the property of a state transition system that states can be rewritten in more than one way yielding the same result. Although it is a desirable property, confluence is often too strict in practical applications because it also considers states that can never be reached in practice. Additionally, sometimes states that have the same semantics in the practical context are considered as different states due to different syntactic representations. By introducing suitable invariants and equivalence relations on the states, programs may have the property to be confluent modulo the equivalence relation w.r.t. the invariant which often is desirable in practice.

In this paper, a sufficient and necessary criterion for confluence modulo equivalence w.r.t. an invariant for Constraint Handling Rules (CHR) is presented. It is the first approach that covers invariant-based confluence modulo equivalence for the de facto standard semantics of CHR. There is a trade-off between practical applicability and the simplicity of proving a confluence property. Therefore, a better manageable subset of equivalence relations has been identified that allows for the proposed confluence criterion and and simplifies the confluence proofs by using well established CHR analysis methods.
\end{abstract}

%% file: confmodeq-body.tex
\section{Introduction}

In program analysis, the \emph{confluence} property of a program plays an important role. It ensures that any computation for a given start state results in the same final state. Hence, if more than one rule is applicable in a state it does not matter which rule is chosen. 

\emph{Constraint Handling Rules (CHR)} is a declarative programming language that has its origins in constraint logic programming \cite{fru_chr_book_2009}. Confluence analysis has been studied for CHR for a long time \cite{abd_fru_meuss_confluence_1996,abd_sem_conf_prop_rules_cp97,abdennadher1999}.

While it is a desirable property, in practical applications confluence is often too strict. For instance, it requires even states that can never be reached in a practical context to satisfy the confluence property. Therefore, \emph{invariant-based confluence} \cite{duck_stuck_sulz_observable_confluence_iclp07,duck_stuckey_sulzmann_observable_confl_chr06,raiser_phdthesis10} has been established. It only considers states that satisfy a user-defined invariant, 
whereas standard confluence analysis considers even states that are invalid and cannot appear at runtime.
With invariant-based confluence analysis it is possible to exclude those states from the confluence analysis as long as the rules of the program maintain the invariant.

Another method of making the confluence property available for more practical programs is to define an equivalence relation on states. A program is \emph{confluent modulo a (user-defined) equivalence relation} if all states in the same equivalence class lead to final states of the same equivalence class \cite{christiansen_2015,christiansen_2017}. In many programs, some states can be considered as equivalent with respect to a user-defined equivalence relation, although their actual representation in the program differs. For example, if sets of numbers are represented as lists, all states with permutations of the same list represent the same set and it might be reasonable to consider them equivalent. Hereby, confluence modulo equivalence can be used to show that for the same start state a program yields the same set as a result, although the actual representation as a list might differ.

There is a trade-off between the applicability in practical contexts and the simplicity of proving a confluence property: There is a  decidable, sufficient and necessary criterion for strict confluence of terminating CHR programs \cite{fru_chr_book_2009}. When adding invariants, decidability of the criterion is lost depending on the invariant. For confluence modulo equivalence, the proofs even become harder as all states in the same equivalence class have to be considered.

In this paper, a sufficient and necessary criterion for invariant-based confluence modulo a user-defined equivalence is presented. For this purpose, a subset of well-behaving equivalence relations is identified for which the proposed criterion can be applied. The confluence criterion is directly available for the equivalence-based operational semantics of CHR \cite{raiser_betz_fru_equivalence_revisited_chr09,raiser_phdthesis10} that is the de facto standard of CHR semantics. By a running example it is shown that the defined subset of equivalence relations is meaningful in a sense that there is a non-trivial equivalence relation that satisfies the restrictions of the subset. Further examples have been tried indicating that the approach is promising to be more widely applicable.

In our approach, we use CHR in its pure form. We then restrict the equivalence relations to a meaningful subset and present a formal proof method for invariant-based confluence modulo equivalence.

The contributions of the paper are
\begin{itemize}
 \item the identification of a subset of equivalence relations (called \emph{compatible} equivalence relations) that maintains the monotonicity property of CHR and therefore allows for a confluence analysis based on rule states and overlaps of rules (c.f. \Cref{sec:equivalence_relations}),
 \item a sufficient and necessary criterion for an invariant-based confluence modulo equivalence for terminating CHR programs with a decidable invariant and a compatible equivalence relation (c.f. \Cref{sec:testing_conf_mod_eq}), and
 \item the application of this approach in a non-trivial running example. 
\end{itemize}

Our approach is the first that covers invariant-based confluence modulo equivalence for the standard semantics of CHR. Other approaches either only consider invariants without user-defined equivalence relations \cite{duck_stuck_sulz_observable_confluence_iclp07,duck_stuckey_sulzmann_observable_confl_chr06,raiser_phdthesis10} or use a special-purpose operational semantics of CHR that is claimed to extend the standard semantics \cite{christiansen_2017,christiansen_2015}. 
The latter approach introduces a meta-level to prove confluence modulo equivalence.
In contrast to the meta-level proof method, the confluence criterion in this paper uses well-established standard notions of CHR states and analysis methods. 

The paper is structured as follows: In \Cref{sec:preliminaries} the preliminaries necessary for understanding the paper are given. For this purpose, definitions of confluence modulo equivalence, Constraint Handling Rules and some program analysis methods for CHR are recapitulated. Then, the subset of equivalence relations regarded in this paper is defined in \Cref{sec:equivalence_relations}. The proof method for invariant-based confluence modulo equivalence is given in \Cref{sec:testing_conf_mod_eq}. The results and its relation to existing work are discussed in \Cref{sec:discussion_related_work}.

\section{Preliminaries}
\label{sec:preliminaries}

We recapitulate the basic notions of confluence modulo equivalence, give a brief introduction to CHR and some program analysis techniques and summarize the established results for (invariant-based) confluence in CHR.

\subsection{Confluence Modulo Equivalence}
\label{sec:preliminaries:conf_mod_eq}

The notion of confluence modulo equivalence is defined for general state transition systems in this section.
\begin{definition}[state transition system]
A \emph{state transition system} is a tuple $(\Sigma,\mapsto)$ where $\Sigma$ is an arbitrary (possibly infinitely large) set of \emph{states} and $\mapsto \subseteq \Sigma \times \Sigma$ is a \emph{transition relation} over the states. By $\mapsto^*$ we denote the reflexive transitive closure of $\mapsto$.
\end{definition}

Informally, confluence modulo equivalence means that all possible computations in a transition system starting in equivalent states finally lead to equivalent states again. We then call two states from those different computations joinable. This is illustrated in \Cref{fig:conf_mod_eq}.

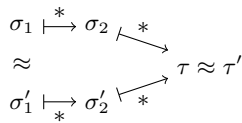
\begin{figure}[htb]
\tikzstyle{trans}+=[draw,|->]
\centering
\begin{tikzpicture}[node distance=1cm]
\node (s1) {$\sigma_1$}; 
\node[below of=s1, node distance=0.5cm] (eq1) {$\approx$}; 
\node[below of=eq1,node distance=0.5cm] (s1p) {$\sigma_1'$}; 
\node[right of=s1] (s2) {$\sigma_2$}; 
\node[right of=s1p] (s2p) {$\sigma_2'$}; 
\node[right  of=eq1, node distance=2.5cm] (tau) {$\tau \approx \tau'$}; 
\path[trans] (s1) -- node[above]{$*$} (s2);
\path[trans] (s1p) -- node[below]{$*$} (s2p);
\path[trans] (s2) -- node[above]{$*$} (tau);
\path[trans] (s2p) -- node[below]{$*$} (tau);
\end{tikzpicture}
\caption{Confluence modulo equivalence}
\label{fig:conf_mod_eq}
\end{figure}

\begin{definition}[joinability modulo equivalence]
 In a state transition system $(\Sigma, \mapsto)$ two states $\sigma, \sigma' \in \Sigma$ are \emph{joinable modulo an equivalence relation $\eq$} if and only if $\exists \tau, \tau' \in \Sigma \enspace . \enspace \sigma \mapsto^* \tau \land \sigma' \mapsto^* \tau' \land \tau \eq \tau'$. We then write $\sigma \joins^\eq \sigma'$. If $\eq$ is the identity equivalence relation $=$, we write $\sigma \joins \sigma'$ and say that $\sigma$ and $\sigma'$ are \emph{joinable}.
\end{definition}

\begin{definition}[confluence modulo equivalence \cite{huet_confluent_1980}]
A state transition system $(\Sigma, \mapsto)$ is confluent modulo an equivalence relation $\eq$, if and only if for all $\sigma_1, \sigma_1', \sigma_2, \sigma_2': (\sigma_1 \eq \sigma_1') \land (\sigma_1 \mapsto^* \sigma_2) \land (\sigma_1' \mapsto^* \sigma_2') \rightarrow (\sigma_1' \joins^\eq \sigma_2')$.
\end{definition}
If $\eq$ is the state equivalence relation $=$, confluence modulo $=$ coincides with basic confluence. For terminating transition systems, it suffices to show \emph{local confluence}, as we will see in the following definition and theorem. 
\begin{definition}[local confluence \cite{huet_confluent_1980}]
\label[definition]{def:alpha_beta_property}
A state transition system $(\Sigma, \mapsto)$ has the $\alpha$ and $\beta$ property w.r.t. an equivalence relation $\eq$ if and only if it satisfies the $\alpha$ and $\beta$ conditions, respectively:
\begin{description}
 \item[$\alpha$:] $\forall \sigma, \tau, \tau' \in \Sigma : \sigma \mapsto \tau \land \sigma \mapsto \tau' \rightarrow \tau \joins^\eq \tau'$.
 \item[$\beta$:] $\forall \sigma, \tau, \tau' \in \Sigma : \sigma \mapsto \tau \land \sigma \eq \tau' \rightarrow \tau \joins^\eq \tau'$.
\end{description}
A state transition system is \emph{locally confluent modulo an equivalence relation $\eq$} if and only if it has the $\alpha$ and the $\beta$ property.
\end{definition}
In the theorem of Huet \cite{huet_confluent_1980} it is shown that local confluence modulo an equivalence relation $\eq$ implies confluence modulo $\eq$ for terminating transition systems.
\begin{theorem}[Huet \cite{huet_confluent_1980}]
\label{thm:huet}
Let $(\Sigma,\mapsto)$ be a terminating transition system. For any equivalence $\eq$, $(\Sigma,\mapsto)$ is confluent modulo $\eq$ if and only if $(\Sigma,\mapsto)$ is locally confluent modulo $\eq$. 
\end{theorem}

\subsection{Constraint Handling Rules}
\label{sec:preliminaries:chr}

We now define the state transition system of CHR. We begin with CHR states.
\begin{definition}[CHR state]
A \emph{CHR state} is a tuple $\langle \mathbb{G} ; \mathbb{C} ; \mathbb{V} \rangle$ where the \emph{goal} $\mathbb{G}$ is a multi-set of constraints, the \emph{built-in constraint store} $\mathbb{C}$ is a conjunction of built-in constraints and $\mathbb{V}$ is a set of \emph{global variables}. All variables occurring in a state that are not global are called \emph{local variables} \cite[p. 33 et seq., def. 8.1]{raiser_phdthesis10}. If the contents of $\mathbb{G}$ and $\mathbb{V}$ are empty or clear from the context, we might omit them. 
\end{definition}

CHR states can be modified by rules that together form a CHR program. 
\begin{definition}[CHR program]
A \emph{CHR program} is a finite set of so-called simpagation rules of the form $r : H_\mathrm{k} ~\backslash~ H_\mathrm{r} \Leftrightarrow G ~|~ B_\mathrm{c} , B_\mathrm{b}$
 where $r$ is an optional rule name, the heads $H_\mathrm{k}$ and $H_\mathrm{r}$ are multi-sets of CHR constraints, the guard $G$ is a conjunction of built-in constraints and the body is a multi-set of CHR constraints $B_\mathrm{c}$ and a conjunction of built-in constraints $B_\mathrm{b}$. If $G$ is empty, it is interpreted as the built-in constraint $\top$. 
 
We introduce short forms for the following special cases :
\begin{description}
 \item[Simplification Rules] If $H_\mathrm{k} = \emptyset$, we also write $H_\mathrm{r} \Leftrightarrow G ~|~ B_\mathrm{c} , B_\mathrm{b}$.
 \item[Propagation Rules] If $H_\mathrm{r} = \emptyset$, we also write $H_\mathrm{k} \Rightarrow G ~|~ B_\mathrm{c}, B_\mathrm{b}$.
\end{description}
\end{definition}

Informally, a rule is applicable, if the heads match constraints from the goal store $\mathbb{G}$ and the guard holds, i.e. is a consequence of the built-in constraints $\mathbb{C}$. In that case, the state is rewritten: The constraints matching the part $H_\mathrm{r}$ of the head are removed and the constraints matching $H_\mathrm{k}$ are kept. The user-defined body constraints $B_\mathrm{c}$ are added to the goal store $\mathbb{G}$, the built-in body constraints $B_\mathrm{b}$ and the constraints from the guard $G$ are added to the built-in store $\mathbb{C}$.

\begin{example}[Multi-Set Items \cite{christiansen_2017}]
\label{ex:set_items}
Consider the following small CHR program, that collects items represented in individual $\constr{item}/1$ constraints to a multi-set represented by a constraint of the form $\constr{mset}(L)$ where $L$ is a list of items. The program has the following rule:
\begin{equation*}
\constr{mset}(L), \constr{item}(A) \Leftrightarrow \constr{mset}([A|L]).
\end{equation*}
For the initial constraint store $\constr{item}(a), \constr{item}(b), \constr{mset}([])$ the program can apply the rule on $\constr{mset}([])$ and $\constr{item}(b)$ which results in the constraint store $\constr{set}([b]), \constr{item}(a)$. The rule can be applied again to this state,  resulting in the constraint store $\constr{mset}([a,b])$. However, the same program can also yield the constraint store $\constr{mset}([b,a])$. Hence, the program is not confluent. In the following, we will go back to this running example and provide an invariant and equivalence relation together with a proof method to show that the program is actually confluent modulo the equivalence relation w.r.t. the invariant.
\end{example}

In the context of the operational semantics, we assume a constraint theory $\mathcal{CT}$ for the interpretation of the built-in constraints. We define an equivalence relation over CHR states.
\begin{definition}[state equivalence \cite{raiser_phdthesis10,raiser_betz_fru_equivalence_revisited_chr09}]
\label[definition]{def:chr_state_equiv}
Let $\rho_i := \langle \mathbb{G}_i ; \mathbb{C}_i ; \mathbb{V}_i \rangle$ for $i = 1,2$ be two CHR states with local variables $\bar{y_1},\bar{y_2}$ that have been renamed apart. $\rho_1 \equiv \rho_2$ if and only if 
$\mathcal{CT} \models \forall(\mathbb{C}_1 \rightarrow \exists \bar{y_2}.((\mathbb{G}_1 = \mathbb{G}_2) \land \mathbb{C}_2)) \land \forall(\mathbb{C}_2 \rightarrow \exists \bar{y_1}.((\mathbb{G}_1 = \mathbb{G}_2) \land \mathbb{C}_1))$ where $\forall F$ is the universal closure of formula $F$.
The equivalence class of a CHR state is defined as as $[\rho] := \{ \rho' ~|~ \rho' \equiv \rho \}$.
\end{definition}

\begin{example}[state equivalence]
By the above definition of state equivalence, the following states are equivalent \cite[p. 34]{raiser_phdthesis10}:
\begin{itemize}
 \item $\langle c(X); \top ; \emptyset \rangle \equiv \langle c(Y); \top ; \emptyset \rangle$, i.e. local variables can be renamed.
 \item $\langle c(X); X{=}0 ; \{ X \} \rangle \equiv \langle c(0); X{=}0 ; \{X\} \rangle$, i.e. variable bindings from the built-in store can be applied to the goal store.
 \item $\langle \emptyset; X{=}Y \land Y{=}0 ; \emptyset \rangle \equiv \langle \emptyset; X{=}0 \land Y{=}0 ; \emptyset \rangle$, i.e. equivalent built-in stores can be interchanged.
 \item $\langle c(0); \top ; \{X\} \rangle \equiv \langle c(0); \top ; \emptyset \rangle$, i.e. unused global variables can be omitted.
 \item However, $\langle c(X); \top ; \{X\} \rangle \not\equiv \langle c(Y); \top ; \{Y\} \rangle$, i.e. $X$ and $Y$ are free variables and therefore the logical readings of the states are different. Global variables can be used to bridge information between two states. 
\end{itemize}
\end{example}

The operational semantics is now defined by the following transition scheme over equivalence classes of CHR states.
\begin{definition}[operational semantics of CHR \cite{raiser_phdthesis10,raiser_betz_fru_equivalence_revisited_chr09}]
\label[definition]{def:chr_operational_semantics}
For a CHR program, the state transition system over CHR states and the rule transition relation $\mapsto$ is defined as the following transition scheme:
\begin{equation*}
\frac
 {
   r : H_\mathrm{k} ~\backslash~ H_\mathrm{r} \Leftrightarrow G ~|~ B_\mathrm{c} , B_\mathrm{b}
 }
 {
 [\langle H_\mathrm{k} \uplus H_\mathrm{r} \uplus \mathbb{G} ; G \land \mathbb{C} ; \mathbb{V} \rangle ] \mapsto^r [ \langle H_\mathrm{k} \uplus B_\mathrm{c} \uplus \mathbb{G} ; G \land B_\mathrm{b} \land \mathbb{C} ; \mathbb{V} \rangle ]
 }
\end{equation*}
Thereby, $r$ is a variant of a rule in the program such that its local variables are disjoint from the variables occurring in the representative of the pre-transition state. We may just write $\mapsto$ instead of $\mapsto^r$ if the rule $r$ is clear from the context.
\end{definition}
From now on, we only consider equivalence classes of CHR states, since the state transition system is defined over equivalence classes.

An important analysis technique is the merging of states.
\begin{definition}[merge operator $\merge$]
\label[definition]{def:merge}
Let $\sigma_i = \langle \mathbb{G}_i ; \mathbb{B}_i ; \mathbb{V}_i \rangle$ for $i = 1, 2$ be two CHR states such that local variables of one state are disjoint from all variables in the other state. Then for a set $\mathbb{V}$ of variables
\begin{equation*}
\sigma_1 \merge_\mathbb{V} \sigma_2 := \langle \mathbb{G}_1 \uplus \mathbb{G}_2 ; \mathbb{B}_1 \land \mathbb{B}_2 ; (\mathbb{V}_1 \cup \mathbb{V}_2) \setminus \mathbb{V} \rangle.
\end{equation*}
For equivalence classes of CHR states, the merging is defined as $[\sigma_1] \merge_\mathbb{V} [\sigma_2] := [\sigma_1 \merge_\mathbb{V} \sigma_2]$
for two representatives of the equivalence class that have disjoint variables. For $\mathbb{V} = \emptyset$ we write $[\sigma_1] \merge [\sigma_2]$
\cite[p. 50, def. 10.1]{raiser_phdthesis10}.
\end{definition}

Since local variables have to be disjoint when merging two states, it is not possible to extract information about them directly. For instance, $[\langle c(X); X{=}1; \emptyset \rangle]$ is the version of $[\langle c(X); \top; \emptyset \rangle]$ with the local variable $X$, where $X$ is bound to the number 1. In the state $[\langle c(X), X{=}1, \emptyset \rangle]$, we would consider $X{=}1$ as contextual information about the local variable $X$. It is not possible to extract this information by $[\langle c(X); \top; \emptyset\rangle] \merge [\langle \emptyset; X{=}1; \emptyset \rangle]$,
since $[\langle \emptyset; X{=}1; \emptyset \rangle] = [\langle \emptyset; \top; \emptyset\rangle] = [\sigma_\emptyset]$, i.e. the empty state. Hence, the result of merging the two states is $[\langle c(X); \top; \emptyset\rangle]$ although we would like to see the result $[\langle c(X); X{=}1; \emptyset \rangle]$.

It is necessary to rather make $X$ a global variable first that is reduced by the merge operator $\merge_{\{X\}}$:
\begin{equation*}
   [\langle c(X); \top; \{X\} \rangle] \merge_{\{X\}} [\langle\emptyset; X{=}1; \{X\} \rangle] = [\langle c(X); X{=}1; \emptyset \rangle] = [\langle c(1); \top; \emptyset \rangle].
\end{equation*}
Global variables can thus be used to share information between two states that are merged. \cite[p. 50, ex. 10.2]{raiser_phdthesis10} In the confluence criterion in \Cref{sec:testing_conf_mod_eq}, we only generate states from the program source code where all variables are global. 

In general, $\merge_\mathbb{V}$ is not associative. However, the following lemma shows a restricted form of associativity that is used in the proof of the confluence modulo equivalence criterion in \Cref{sec:testing_conf_mod_eq}.
\begin{lemma}
\label[lemma]{lemma:special_assoc_merge}
Let $\sigma_1, \sigma_2, \sigma_3$ be CHR states such that no local variable of a state occurs in another state. Then $[\sigma_1] \merge_\mathbb{V} ([\sigma_2] \merge [\sigma_3]) = ([\sigma_1] \merge [\sigma_2]) \merge_\mathbb{V} [\sigma_3]$ holds for all $\mathbb{V}$ \cite[p. 52, lemma 10.7]{raiser_phdthesis10}.
\end{lemma}

\subsection{Confluence of CHR Programs}
\label{sec:confluence_chr}

The idea of the confluence criterion is to exploit the \emph{monotonicity} property of CHR, i.e. that all rules applicable in one state are applicable in any larger state. 
\begin{lemma}[monotonicity] 
\label[lemma]{lemma:monotonicity}
If $[\sigma] \mapsto [\tau]$, then $[\sigma] \merge_\mathbb{V} [\sigma'] \mapsto [\tau] \merge_\mathbb{V} [\sigma']$ for all $\mathbb{V}$ \cite[p. 51, lemma 10.4]{raiser_phdthesis10}.
\end{lemma}

\subsubsection{Basic Confluence Test}

Monotonicity allows us to reason from states about larger states. The idea of the basic confluence test is to construct a finite set of \emph{rule states} that consist of the head and guard constraints of a rule and then overlap them with all other rule states. Intuitively, overlapping two rules means that a state is constructed where parts of the rule heads are equated (if possible) and the rest is just included. In such a state, both rules are applicable. 

By applying the overlapping rules to the overlap state, we get a \emph{critical pair}. Thereby, one state is the result after applying the first overlapping rule to the overlap state and the other state is the result after applying the second rule to the overlap state. If all critical pairs are joinable, the program is locally confluent. In the following, we formalize this idea. The definitions are taken from \cite{raiser_phdthesis10}. Similar definitions can be found in \cite{fru_chr_book_2009}.
\begin{definition}[rule state] 
\label[definition]{def:rule_state}
For a rule $r : H_\mathrm{k} ~\backslash~ H_\mathrm{r} \Leftrightarrow G ~|~ B_c, B_b$ let $\mathbb{V}$ be the variables occurring in $H_\mathrm{k}, H_\mathrm{r}$ and $G$. Then the state $\langle H_\mathrm{k} \uplus H_\mathrm{r} ; G ; \mathbb{V} \rangle$ is called the \emph{rule state} of $r$. In the literature, the rule states are sometimes called \emph{minimal states}. 
\cite[p. 78, def. 13.8]{raiser_phdthesis10}
\end{definition}

\begin{definition}[overlap] 
\label[definition]{def:overlap}
For any two (not necessarily different) rules of a CHR program of the form
$r_1 : H_\mathrm{k} ~\backslash~ H_\mathrm{r} \Leftrightarrow G ~|~ B_\mathrm{c} , B_\mathrm{b},$
$r_2 : H_\mathrm{k}' ~\backslash~ H_\mathrm{r}' \Leftrightarrow G' ~|~ B_\mathrm{c}' , B_\mathrm{b}'$
and with variables that are renamed apart, let
 $O_\mathrm{k} \subseteq H_\mathrm{k}, O_\mathrm{r} \subseteq H_\mathrm{r},$
 $O_\mathrm{k}' \subseteq H_\mathrm{k}', O_\mathrm{r}' \subseteq H_\mathrm{r}'$
be subsets of the heads of the rules such that for
$B := ((O_\mathrm{k} \uplus O_\mathrm{r}) = (O_\mathrm{k}' \uplus O_\mathrm{r}')) \land G \land G'$                                                      
it holds that $\mathcal{CT} \models \exists.B$ and $(O_\mathrm{r} \uplus O_\mathrm{r}') \neq \emptyset$, where $\exists.B$ is the existential closure over $B$. Then the state
 $\sigma := \langle K \uplus K' \uplus R \uplus R' \uplus O_\mathrm{k} \uplus O_\mathrm{r} ; B ; \mathbb{V} \rangle$
is called an \emph{overlap} of $r_1$ and $r_2$ where $\mathbb{V}$ is the set of all variables occurring in heads and guards of both rules and
 $K := H_\mathrm{k} \setminus O_\mathrm{k}$, 
 $K' := H_\mathrm{k}' \setminus O_\mathrm{k}'$, 
 $R := H_\mathrm{r} \setminus O_\mathrm{r}$, 
 $R' := H_\mathrm{r}' \setminus O_\mathrm{r}'$.
The pair of states $(\sigma_1, \sigma_2)$ with $\sigma_1 := \langle K \uplus K' \uplus R' \uplus O_\mathrm{k}  \uplus B_\mathrm{c}  ; B \land B_\mathrm{b}  ; \mathbb{V} \rangle$ and $\sigma_2 := \langle K \uplus K' \uplus R  \uplus O_\mathrm{k}' \uplus B_\mathrm{c}' ; B \land B_\mathrm{b}' ; \mathbb{V} \rangle$
is called \emph{critical pair} of the overlap $\sigma$. The critical pair can be obtained by applying the rules to the overlap state.
\cite[p. 82, def. 14.5]{raiser_phdthesis10}
\end{definition}

\subsubsection{Invariant-Based Confluence Test}

The idea to exploit monotonicity fails, when invariants on the states are introduced. A property $\inv$ is an invariant if and only if for all states $[\sigma]$ where $\inv([\sigma])$ holds and for all $[\tau]$ with $[\sigma] \mapsto^* [\tau]$ the invariant $\inv([\tau])$ holds as well.

If in the confluence test a constructed overlap does not satisfy the invariant, then this overlap state is not part of the transition system and therefore no information can be gained from analyzing it. It is not possible to just ignore such states as there are invariants that are not satisfied in an overlap state, but might be satisfied in a larger state. There are also invariants, that cannot be satisfied by state extension (c.f. \Cref{ex:set_items_inv}). 

Nevertheless, the idea to use overlap states for confluence analysis can be generalized, such that it can be used for invariant-based confluence. For this purpose, for an invariant $\inv$ and an overlap state $[\sigma]$ the set of all extensions of $[\sigma]$ such that $\inv$ holds -- denoted by $\Sigma^\inv([\sigma])$ -- is considered. As this set usually is infinitely large, we want to extract a set of minimal elements of $\Sigma^\inv([\sigma])$, called $\mathcal{M}^\inv([\sigma])$, that have to be considered to show local confluence w.r.t. $\inv$. However, for this purpose a partial order on states has to be defined. The set $\mathcal{M}^\inv([\sigma])$ is finite for many invariants, but there are examples of invariants that lead to infinite sets of minimal elements.

In \cite{duck_stuck_sulz_observable_confluence_iclp07,duck_stuckey_sulzmann_observable_confl_chr06,raiser_phdthesis10} the following has been proven: If we can show that for all overlap states $[\sigma]$ of a terminating program the critical pairs derived from all states in $\mathcal{M}^\inv([\sigma])$ are joinable, the program is confluent w.r.t. to $\inv$.

We now give formal definitions of the notions used in the above description. Since it is a commutative monoid, a partial order can be derived from the merge operator \cite{raiser_phdthesis10}:
\begin{lemma}[partial order $\ltstate$]
For the set of CHR states $\Sigma$, the relation $\ltstate : \Sigma \times \Sigma$ defined as $[\sigma] \ltstate [\sigma'] \text{ if and only if } \exists [\hat{\sigma}] ~.~ [\sigma] \merge [\hat{\sigma}] = [\sigma']$ where $\sigma, \sigma' \in \Sigma$ is a partial order. \cite[p. 53, lemma 10.8]{raiser_phdthesis10}
\end{lemma}
In \cite{duck_stuck_sulz_observable_confluence_iclp07,duck_stuckey_sulzmann_observable_confl_chr06}, another partial order has been defined. However, it has been shown that the relation defined there is not a partial order by mistake \cite{raiser_phdthesis10}. Therefore, we use the partial order that has first been introduced in \cite[p. 53, lemma 10.8]{raiser_phdthesis10} to avoid these problems. 

For an invariant, we can now define the set of minimal elements that extend a state such that the invariant does hold.
\begin{definition}[minimal elements]
For an invariant $\inv$, let the set $\Sigma^\inv([\sigma]) := \{ [\sigma'] ~|~ \inv([\sigma \merge \sigma']) \land \sigma' \text{ has no local variables } \}$. The set $\mathcal{M}^\inv([\sigma])$ is the set of $\ltstate$-minimal elements of  $\Sigma^\inv([\sigma])$, i.e.
\begin{equation*}
\forall [\sigma'] \in \Sigma^\inv([\sigma]) ~.~ \exists [\sigma_\mathrm{m}] \in \mathcal{M}^\inv([\sigma]) ~.~ [\sigma_\mathrm{m}] \ltstate [\sigma'].
\end{equation*}
\cite[p. 80, def. 13.11]{raiser_phdthesis10}
\end{definition}

Note that for an invariant $\inv$ and a state $[\sigma]$ where $\inv([\sigma])$ holds, the set of minimal extensions is $\mathcal{M}^\inv([\sigma]) = \{ [\sigma_\emptyset] \}$, where $\sigma_\emptyset := \langle \emptyset ; \top ; \emptyset \rangle$ is the empty state \cite[p.80, lemma 13.13]{raiser_phdthesis10}. The invariant-based confluence test then coincides with the basic confluence criterion. In \Cref{sec:testing_conf_mod_eq} we generalize the idea of the invariant-based confluence test for invariant-based confluence modulo equivalence.

\begin{example}[Multi-Set Items (cont.)]
\label{ex:set_items_inv}
For the multi-set program from \Cref{ex:set_items}, the following problem arises: If there is more than one $\constr{mset}$ constraint, the program can choose non-deterministically where to add an item. Therefore, it cannot be confluent. For instance, the following transitions in shorthand notation is possible: $\constr{mset}([a]), \constr{mset}([b]), \constr{item}([c])$ can either end in the final state  $\constr{mset}([a,c]), \constr{mset}([b])$ or $\constr{mset}([a]), \constr{mset}([b,c])$.

The problem can be solved by introducing the multi-set invariant $\mathcal{S}$: In every CHR state there is at most one $\constr{mset}(\_)$ constraint. Note that the set of minimal extensions $\mathcal{M}^{\mathcal{S}}([\sigma]) = \emptyset$ for all states $[\sigma]$, as there are no extensions for states that do not satisfy the invariant (i.e. where there is more than one $\constr{mset}$ constraint) such that the invariant is satisfied (i.e. there is at most one $\constr{mset}$ constraint).
\end{example}

\section{Compatibility of Equivalence Relations}
\label{sec:equivalence_relations}

In this section, we motivate a restriction of equivalence relations that make confluence modulo equivalence analysis manageable. Note that in the context of confluence modulo equivalence, typically user-defined equivalence relations different from state equivalence (c.f. \Cref{def:chr_state_equiv}) are regarded. State equivalence is referred to by $\equiv$ or by the corresponding equivalence class brackets $[\cdot]$. The symbol $\eq$ denotes some general user-defined equivalence relation that is potentially different from $\equiv$ (but is not required to).

In the confluence criterion, we want to use the idea to exploit monotonicity of CHR to reason from small states that come from the rules in the program over all states. However, monotonicity can be broken by user-defined equivalence relations. This means that in general for two states with $[\sigma] \eq [\sigma']$, it is possible that an extension with $[\tau] \eq [\tau']$ leads to states that are not equivalent, i.e.  $[\sigma] \merge_\mathbb{V} [\tau] \noteq [\sigma'] \merge_\mathbb{V} [\tau']$ as shown in the following example.

\begin{example}
We construct an equivalence relation that breaks monotonicity. Let $\#c : \Sigma \rightarrow \mathbb{N}_0$ be a function that returns the number of constraints $c$ in the goal store of a state.  We separate the CHR state space $\Sigma$ into two disjoint subsets:
\begin{align*}
\Sigma_1 :=& \{ [\sigma] ~|~ \#c([\sigma]) < 3 \}, & &\Sigma_2 :=& \{ [\sigma] ~|~ \#c([\sigma]) \geq 3 \}.
\end{align*}
The partition of the state space clearly defines an equivalence relation $\eq$ with equivalence classes $\Sigma_1$ and $\Sigma_2$.

Let $[\sigma_1] = [\langle c ; \top ; \emptyset \rangle]$ and $[\sigma_2] = [\langle c, c ; \top ; \emptyset \rangle]$. Since $[\sigma_1], [\sigma_2] \in \Sigma_1$, it holds that $[\sigma_1] \eq [\sigma_2]$. Let $[\tau] = [\langle c ; \top ; \emptyset \rangle]$. If we extend the two states by $[\tau]$, the extended states are not equivalent any more:
\begin{equation*}
[\sigma_1] \merge [\tau] = [\langle c, c ; \top ; \emptyset \rangle] \in \Sigma_1\text{, but }[\sigma_2] \merge [\tau] = [\langle c, c, c ; \top ; \emptyset \rangle] \in \Sigma_2.
\end{equation*}
Hence, although $[\sigma_1] \eq [\sigma_2]$, $[\sigma_1] \merge [\tau] \noteq [\sigma_2] \merge [\tau]$. 

This case does not harm testing for the $\beta$ property, since the extended states do not have to be tested for joinability modulo equivalence according to the $\beta$ property. However, we can construct the converse case: Let $[\sigma_3] = [\langle c, c, c ; \top ; \emptyset \rangle] \in \Sigma_2$. Then $[\sigma_2] \noteq [\sigma_3]$. However, if the two states are extended by $[\tau]$, we get
\begin{align*}
[\sigma_2] \merge [\tau] &= [\langle c, c, c ; \top ; \emptyset \rangle] \in \Sigma_2 \text{, and } & 
[\sigma_3] \merge [\tau] &= [\langle c, c, c, c ; \top ; \emptyset \rangle] \in \Sigma_2.
\end{align*}
Hence, $[\sigma_2] \merge [\tau] \eq [\sigma_3] \merge [\tau]$, although $[\sigma_2] \noteq [\sigma_3]$. This is critical to the $\beta$ property: Now it is not possible any more to use a rule state and its equivalent states to reason about all states as we miss some larger state by this attempt.
\end{example}

To ensure monotonicity in the context of equivalence relations, we need equivalence to be maintained by the merge operator. The equivalence relation is then called a \emph{congruence relation} with respect to the merge operator.
\begin{definition}[congruence relation]
\label[definition]{def:congruence_relation}
An equivalence relation $\eq \subseteq A \times A$ is called a \emph{congruence relation} with respect to an operator $\circ : A \times A \rightarrow A$ if for all $x, x', y, y'$: If $x \eq x'$ and $y \eq y'$ then $x \circ y \eq x' \circ y'.$
\end{definition}

Unfortunately, this does not suffice to reason from rule states about any other state. It must be ensured that if two states $[\sigma]$ and $[\sigma']$ are equivalent and $[\sigma]$ can be decomposed into two parts, then $[\sigma']$ must be decomposable into two parts that are equivalent to the decomposition of $[\sigma]$. This ensures that when showing joinability of two small states, the larger states can still be joined, as they are syntactically decomposable into smaller joinable states.
\begin{definition}[split property]
\label[definition]{def:split_property}
An equivalence relation $\eq \subseteq A \times A$ has the \emph{split property} with respect to an operator $\circ : A \times A \rightarrow A$ if for all $x, x_1, x_2, y$: If $x = x_1 \circ x_2$ and $x \eq y$ then $\exists y_1, y_2$ such that $x_1 \eq y_1, x_2 \eq y_2$ and $y = y_1 \circ y_2$.
\end{definition}
The split property assumes a syntactic relation between two states that are equivalent under an equivalence relation. If a state can be split into two parts and is equivalent to another state, this state can be split into equivalent parts.

\begin{example}
This example defines an equivalence relation $\hat{=}$ that does not satisfy the split property. It is the smallest equivalence relation where the following two conditions hold: If $\sigma \equiv \sigma'$ then also $\sigma \hat{=} \sigma'$. Additionally, if $\langle \mathbb{G} ; \mathbb{B} ; \mathbb{V} \rangle \hat{=} \langle \mathbb{G}' ; \mathbb{B}' ; \mathbb{V}' \rangle$, then $\langle \{c,c\} \uplus \mathbb{G} ; \mathbb{B} ; \mathbb{V} \rangle \hat{=} \langle \{ d \} \uplus \mathbb{G}' ; \mathbb{B}' ; \mathbb{V}' \rangle$. Hence, all pairs of $c$ constraints can be replaced by a $d$ constraint. 

The equivalence relation obviously is a congruence relation w.r.t. $\merge$. However, it does not have the split property: Let $\sigma \equiv c,c$ be a CHR state in shorthand notation. Then $\sigma \equiv c \merge c$. By definition of $\hat{=}$, we have that $\sigma \hat{=} d$. However, there are no $\sigma_1,\sigma_2$ such that $\sigma_1 \hat{=} c$, $\sigma_2 \hat{=} c$ and $d \equiv \sigma_1 \merge \sigma_2$.

In the confluence test, for all states $\sigma$ it has to be shown that if $\sigma \equiv \sigma'$ and $\sigma \mapsto_r \tau$ then $\sigma' \joins^\eq \tau$ to satisfy the $\beta$ property. By the application of $r$ to $\sigma$, we know that for the rule state $\sigma_r$, $\sigma$ can be split into $[\sigma] = [\sigma_r] \merge [\delta]$. To reason from joinability of $\sigma_r$ and all its equivalent states, we also have to be able to split $\sigma'$ into two parts $[\sigma_r']$ and $[\delta']$. However, for the congruence relation $\hat{=}$ this is not possible as we have shown before. Hence, to use the idea of reasoning from rule states about all larger states, the equivalence relation has to be $\merge$-compatible.

Note that the split property is only required to hold for states where the invariant holds. Hence, by an appropriate invariant, the split property can be recovered to show confluence w.r.t. this invariant.
\end{example}

\begin{definition}[compatibility]
\label[definition]{def:compatibility}
An equivalence relation $\eq$ is $\circ$-\emph{compatible} w.r.t. an operator $\circ$ if it is a congruence relation with the split property w.r.t. $\circ$.
\end{definition}
At first glance, compatibility is a strict property that does not seem to be satisfied by many equivalence relations. However, there are interesting $\merge$-compatible equivalence relations different from the trivial state equivalence: 

\begin{example}[Multi-Set Items (cont.)]
\label{ex:set_items_seteq}
\Cref{ex:set_items} is continued by introducing the following equivalence relation $\seteq$ that is the smallest equivalence relation on CHR states such that
$[\langle \{ \constr{mset}(S_1) \} \uplus \mathbb{G}_1 ; \mathbb{B}_1 ; \mathbb{V}_1 \rangle ] \seteq [\langle \{ \constr{mset}(S_2) \} \uplus \mathbb{G}_2 ; \mathbb{B}_2 ; \mathbb{V}_2 \rangle ]$
if and only if $S_1$ is a permutation of $S_2$ and $[\langle \mathbb{G}_1 ; \mathbb{B}_1 ; \mathbb{V}_1 \rangle ] \seteq [\langle \mathbb{G}_2 ; \mathbb{B}_2 ; \mathbb{V}_2 \rangle ].$

For instance, the following states in shorthand notation are equivalent according to $\seteq$: $\constr{set}([a,b]), \constr{set}([c,d]), \constr{item}(e) \seteq \constr{set}([b,a]), \constr{set}([d,c]), \constr{item}(e)$ and $\constr{item}(a) \seteq \constr{item}(a)$. However, $\constr{set}([a,b]), \constr{item}(c) \nseteq \constr{set}([a,b]),\constr{item}(d)$ and $\constr{set}([a,b]), \constr{item}(c) \nseteq \constr{set}([a,b]),\constr{item}(c),\constr{item}(c)$ because the second item $c$ does not have a partner in the first state. Similarly,   $\constr{set}([a,b]), \constr{set}([b,a]) \nseteq \constr{set}([a,b])$ because there is only one $\constr{set}$ constraint on the right hand side. 

Note that by this definition the following holds for states with unbound variables: $[\langle \constr{mset}(X) ; \constr{perm}(X,Y) ; \{ X,Y \} \rangle ] \seteq [\langle \constr{mset}(Y) ; \constr{perm}(X,Y) ; \{ X,Y \} \rangle ]$ where $\constr{perm}(X,Y)$ is a built-in constraint that is true, if $X$ is a permutation of $Y$, but  $[\langle \constr{mset}(X) ; \top ; \{ X\} \rangle ] \nseteq [\langle \constr{mset}(Y) ; \top ; \{ Y\} \rangle ]$. The two variables $X$ and $Y$ are free variables and therefore it is not clear that they are permutations of each other. By adding that $X$ is a permutation of $Y$, the two states are equivalent.

This equivalence relation is $\merge$-compatible. For reasons of space, the proof is provided in \Cref{app:seteq_compatibility_proof}.
\end{example}

\section{Confluence Modulo Equivalence w.r.t. an Invariant} 
\label{sec:testing_conf_mod_eq}

First of all, the notion of invariant-based confluence modulo equivalence is defined.
\begin{definition}[$\inv$-confluence modulo $\eq$]
\label{def:inv_confmeq}
A state transition system is $\inv$-confluent modulo an equivalence relation $\eq$ for an invariant $\inv$ if and only if 
\begin{gather*}
\forall \sigma_1,\sigma_2,\sigma_1',\sigma_2' ~.~ \inv(\sigma_1) \land \inv(\sigma_1') \land \sigma_1 \eq \sigma_1' \land \sigma_1 \mapsto^* \sigma_2 \land \sigma_1' \mapsto^* \sigma_2' \\
\rightarrow \exists \sigma_3, \sigma_3' ~.~ \sigma_2 \mapsto^* \sigma_3 \land \sigma_2' \mapsto^* \sigma_3' \land \sigma_3 \eq \sigma_3'.
\end{gather*}
\end{definition}

In practice, the following restiction is made on invariants:
\begin{definition}[$\eq$ maintains $\inv$]
An invariant $\inv$ is maintained by an equivalence relation $\eq$, if and only if for all states $[\sigma] \eq [\sigma']$ it holds that $\inv([\sigma]) \leftrightarrow \inv([\sigma'])$.
\end{definition}
This restriction ensures the practicability of \Cref{def:inv_confmeq}, since it may be inelegant and misleading if in a program that is $\inv$-confluent modulo $\eq$ there exist two equivalent states where one is part of the program (i.e. the invariant holds) and the other is not. This may have undesired effects in further analysis. Hence, the invariant and equivalence relation should be chosen such that they are compliant anyway, although it is not required by \Cref{def:inv_confmeq}.

The following lemma is an important generalization of the joinability corollary in \cite[p. 85, cor. 14.9]{raiser_phdthesis10} that is a direct consequence of monotonicity. The idea was that if two states are joinable, they are still joinable if they are extended by the identical state. In the context of confluence modulo equivalence, we have to generalize this approach of exploiting monotonicity such that the state extensions are not required to be syntactically identical, but equivalent for some user-defined compatible equivalence relation. 
\begin{lemma}[joinability]
\label[lemma]{lemma:joinability_monotonicity}
Let $\eq$ be a congruence relation with respect to $\merge$ and $[\sigma_1], [\sigma_2], [\sigma_1'], [\sigma_2']$ be CHR states with $[\sigma_1'] \eq [\sigma_2']$. If $[\sigma_1] \joins^\eq [\sigma_2]$ then $([\sigma_1] \merge_\mathbb{V} [\sigma_1']) \joins^\eq ([\sigma_2] \merge_\mathbb{V} [\sigma_2'])$ for all $\mathbb{V}$.
\end{lemma}
\begin{proof}
Let $[\sigma_1], [\sigma_2], [\sigma_1'], [\sigma_2']$ be CHR states with $[\sigma_1'] \eq [\sigma_2']$ and $[\sigma_1] \joins^\eq [\sigma_2]$. Hence, there are CHR states $[\tau], [\tau']$ with $[\tau] \eq [\tau']$ and $[\sigma_1] \mapsto^* [\tau]$ and $[\sigma_2] \mapsto^* [\tau']$. 
Due to monotonicity (c.f. \Cref{lemma:monotonicity}), we have that $([\sigma_1] \merge_\mathbb{V} [\sigma_1']) \mapsto^* ([\tau] \merge_\mathbb{V} [\sigma_1'])$ and $([\sigma_2] \merge_\mathbb{V} [\sigma_2']) \mapsto^* ([\tau'] \merge_\mathbb{V} [\sigma_2'])$. Since $[\sigma_1'] \eq [\sigma_2']$, $[\tau] \eq [\tau']$ and $\eq$ is a congruence relation with respect to $\merge$, we have that $([\tau] \merge_\mathbb{V} [\sigma_1']) \eq ([\tau'] \merge_\mathbb{V} [\sigma_2'])$.
\end{proof}

In the next step, we provide a test for the $\alpha$ property in the context of an invariant. The basic idea is that we gather all overlap states and extend them with a minimal extension such that the invariant does hold. For all those minimal extensions of all overlap states we have to show joinability modulo equivalence. Formally, this leads to the following lemma.
\begin{lemma}[$\alpha$ property test]
\label[lemma]{lemma:alpha_property_test}
Let $\mathcal{P}$ be a CHR program, $\inv$ an invariant, $\eq$ a congruence relation and let $\mathcal{M}^\inv([\sigma])$ be well-defined for all overlaps $\sigma$ of rules in $\mathcal{P}$, then:
$\mathcal{P}$ has the $\alpha$ property with respect to $\inv$ and $\eq$ if and only if for all overlaps $\sigma$ with critical pairs $(\sigma_1,\sigma_2)$ and all $[\sigma_\mathrm{m}] \in \mathcal{M}^\inv([\sigma])$ holds $([\sigma_1] \merge [\sigma_\mathrm{m}] \joins^\eq [\sigma_2] \merge [\sigma_\mathrm{m}])$.
\end{lemma}
\begin{proof}
For reasons of space, the proof is provided in \Cref{app:alpha_property_test_proof}. It is similar to the proof for the $\beta$ property.
\end{proof}

To prove local confluence modulo equivalence, we also have to prove the $\beta$ property, i.e. we have to consider that if two states are equivalent, they have to be joinable modulo equivalence. In the following lemma, we adapt the test for the $\alpha$ property to cover the $\beta$ property. 

The main idea is to reason from rule states, i.e. the head and guard constraints of rules, over all states. For this purpose, all rule states have to be extended by a minimal extension such that the invariant holds. Then all states that are equivalent to these extended rule states have to be shown to be joinable to the extended rule state after the rule has been applied. Unfortunately -- depending on the invariant -- in general there can be infinitely many such equivalent states. However, the idea still simplifies the proof procedure for the $\beta$ property, as only rule states have to be considered in contrast to all states of the transition system.

This is not possible for general equivalence relations, but only for those that are compatible to the merge operator and that maintain the invariant. 
\begin{lemma}[$\beta$ property test]
\label[lemma]{lemma:beta_property_test}
Let $\mathcal{P}$ be a CHR program, $\inv$ an invariant, $\eq$ a $\merge$-compatible equivalence relation that maintains $\inv$ and let $\mathcal{M}^\inv([\sigma])$ be well-defined for all rule states $[\sigma]$ in $\mathcal{P}$, then:
$\mathcal{P}$ has the $\beta$ property with respect to $\inv$ and $\eq$ if and only if for all rule states $[\sigma]$ with successor state $[\sigma_1]$, all $[\sigma_2]$ with $[\sigma] \eq [\sigma_2]$ and all $[\sigma_\mathrm{m}^1] \in \mathcal{M}^\inv([\sigma])$ and all $[\sigma_\mathrm{m}^2] \eq [\sigma_\mathrm{m}^1]$ where $\inv([\sigma_2] \merge [\sigma_\mathrm{m}^2])$ is satisfied, it holds that $([\sigma_1] \merge [\sigma_\mathrm{m}^1]) \joins^\eq ([\sigma_2] \merge [\sigma_\mathrm{m}^2])$.
\end{lemma}
\begin{proof}
``$\Rightarrow$'': 
This follows from \Cref{def:alpha_beta_property} and \Cref{lemma:joinability_monotonicity}, since $[\sigma_\mathrm{m}^1] \eq [\sigma_\mathrm{m}^2]$. 

``$\Leftarrow$'': 
Let $[\sigma], [\sigma_1]$ and $[\sigma_2]$ be CHR states where $\inv([\sigma])$ and $\inv([\sigma_2])$ hold and $[\sigma] \mapsto_r [\sigma_1]$ for some rule $r$ and $[\sigma] \eq [\sigma_2]$. Since a rule is applicable in $[\sigma]$, there is a rule state $\sigma_r = \langle \_ ; \_ ; \mathbb{V} \rangle$ of rule $r$ such that for some $[\delta_1] := [ \langle \mathbb{G} ; \mathbb{B} ; \mathbb{V'} \rangle]$ it holds that $[\sigma] = [\sigma_r] \merge_\mathbb{V} [\delta_1].$ The variables $\mathbb{V}$ from the rule $r$ are not part of $[\sigma]$ and are therefore removed by the merging $\merge_\mathbb{V}$. 

By definition of the rule state (c.f. \Cref{def:rule_state}) and definition of the state transition system (c.f. \Cref{def:chr_operational_semantics}), we also have that there is a state $[\sigma_1']$ such that $[\sigma_r] \mapsto_r [\sigma_1'].$ Due to monotonicity (c.f. \Cref{lemma:monotonicity}) it holds that $[\sigma_1] = [\sigma_1'] \merge_\mathbb{V} [\delta_1].$

Let $[\sigma_2] = [\sigma_2'] \merge_\mathbb{V} [\delta_2]$ be a partition of $[\sigma_2]$ such that $[\sigma_2'] \eq [\sigma_r]$ and $[\delta_2] \eq [\delta_1]$. Such a partition exists since $\eq$ is $\merge$-compatible and $[\sigma] \eq [\sigma_2]$ by precondition.

As $\inv([\sigma])$ holds: $[\delta_1] \in \Sigma^\inv([\sigma_r])$ and therefore $\exists [\sigma_\mathrm{m}^1] \in \mathcal{M}^\inv([\sigma_r]) . [\sigma_\mathrm{m}^1] \ltstate [\delta_1]$. This means that there is a minimal element $[\sigma_\mathrm{m}^1]$ in the set of extensions of the rule state $[\sigma_r]$ that extend $[\sigma_r]$ such that the invariant holds.

It follows by definition of $\ltstate$ that $\exists [\delta_1'] . [\delta_1] = [\sigma_\mathrm{m}^1] \merge [\delta_1']$ and hence $[\sigma] = [\sigma_r] \merge_\mathbb{V} ([\sigma_\mathrm{m}] \merge [\delta_1'])$. By \Cref{lemma:special_assoc_merge}, we get $[\sigma] = ([\sigma_r] \merge [\sigma_\mathrm{m}]) \merge_\mathbb{V} [\delta_1']$. Analogously, by substitution of $[\delta_1]$ in $[\sigma_1]$ and due to the split property of $\eq$, we find that $[\sigma_i] = ([\sigma_i'] \merge [\sigma_\mathrm{m}^i]) \merge_\mathbb{V} [\delta_i']$ for $i = 1, 2$ where $[\sigma_\mathrm{m}^1] \eq [\sigma_\mathrm{m}^2]$. 

Since $\inv$ is maintained by $\eq$, we have by precondition that $([\sigma_1'] \merge [\sigma_\mathrm{m}^1]) \joins^\eq ([\sigma_2'] \merge [\sigma_\mathrm{m}^2])$. Since $[\sigma_1] = ([\sigma_1'] \merge [\sigma_\mathrm{m}^1]) \merge_\mathbb{V} [\delta_1']$ and $[\sigma_2] = ([\sigma_2'] \merge [\sigma_\mathrm{m}^2]) \merge_\mathbb{V} [\delta_2']$ and $[\delta_1'] \eq [\delta_2']$, we have by \Cref{lemma:joinability_monotonicity} also that $([\sigma_1] \joins^\eq [\sigma_2])$.
\end{proof}

\begin{theorem}[confluence modulo $\eq$ w.r.t. invariant]
\label{thm:conf_mod_eq}
Let $\inv$ be an invariant and $\mathcal{P}$ an $\inv$-terminating CHR program. $\mathcal{P}$ has the $\alpha$ and $\beta$ property with respect to $\inv$ and an equivalence relation $\eq$ if and only if $\mathcal{P}$ is $\inv$-confluent modulo $\eq$.
\end{theorem}
\begin{proof}
\Cref{thm:huet} is used on the reduced state transition system that only contains states where the invariant holds.
\end{proof}

Note that for testing the $\alpha$ property, the criterion only assumes a congruence relation, whereas for proving the $\beta$ property the split property must hold as well and the invariant must maintain equivalence.

\begin{example}[Item Sets (cont.)]
\label{ex:set_items_confluence}
It is shown that the program from \Cref{ex:set_items} is $\mathcal{S}$-confluent modulo $\seteq$.
\begin{description}
 \item[$\alpha$ property]
 
 The only overlap that satisfies the invariant $\mathcal{S}$ has the shorthand notation
$\constr{item}(A), \constr{item}(B), \constr{set}(L)$. It yields the critical pair $\constr{item}(B), \constr{set}([A|L])$ and $\constr{item}(A), \constr{set}([B|L])$. It can be reduced to $\constr{set}([B,A|L]) \seteq \constr{set}([A,B|L])$.

\item[$\beta$ property]

 All equivalences to the rule state have the form $[\langle \constr{item}(A), \constr{mset}(L);$ $\top;\{A,L\}\rangle] \seteq [\langle \constr{item}(A), \constr{mset}(L'); \top;\{A,L'\}  \rangle]$ where $L'$ is a permutation of $L$. The preconditions of \Cref{lemma:beta_property_test} are satisfied, since both states satisfy the invariant. The both states reduce to the goal stores $\constr{mset}([A|L])$ and $\constr{mset}([A|L'])$. It is clear that those two final states are equivalent and therefore joinable modulo $\seteq$.
\qed
\end{description}
\end{example}

\section{Discussion and Related Work}
\label{sec:discussion_related_work}

The $\alpha$ property test is decidable for terminating programs as long as the invariant and the equivalence relation (a congruence relation w.r.t. $\merge$) are decidable and the set of minimal extensions is finite. In the $\beta$ property test, the class of states that are equivalent to the rule state may be infinitely large in general. 

In the multi-set example (c.f. \Cref{ex:set_items_confluence}), it can be seen that only one other state has to be considered to show joinability of all states equivalent to the rule state, since the CHR semantics allows for logical variables. In general, there might be more complicated equivalence relations that are more difficult to test.

Confluence modulo equivalence with invariants has been studied for a variant of CHR that includes non-logical built-in constraints \cite{christiansen_2017,christiansen_2015}. 
This approach introduces a meta language for CHR to prove confluence modulo equivalence. It is claimed that the traditional proof methods for confluence in CHR expressed in first-order logic are not sufficient in the context of confluence modulo equivalence, especially with non-logical built-in constraints. The meta-level is claimed to allow proving confluence modulo equivalence for all equivalence relations. It is shown to be useful for programs with non-logical built-in constraints. 

In many cases the analysis of programs with purely logical CHR is desired and invariants and equivalence relations behave in a way that allow for a more direct treatment. In our approach, no meta-level is necessary. It is directly available for the de facto standard of CHR semantics. It seems to us that the example in \cite{christiansen_2017} indicates that for proving confluence modulo equivalence with the meta-level approach, monotonicity and therefore $\merge$-compatibility are used implicitly.

Invariant-based confluence (or observable confluence) for CHR without user-defined equivalence relations has been studied in \cite{duck_stuck_sulz_observable_confluence_iclp07,duck_stuckey_sulzmann_observable_confl_chr06}. In \cite{raiser_phdthesis10}, it has been shown that the proposed partial order is not well-defined. Our approach integrates the corrected version of invariant-based confluence as found in \cite{raiser_phdthesis10}. Additionally, it extends the idea by confluence modulo user-defined equivalence relations.

\section{Conclusion and Future Work}

A sufficient and necessary criterion for confluence modulo equivalence w.r.t. an invariant has been presented and formally proven (c.f. \Cref{lemma:joinability_monotonicity,lemma:alpha_property_test,lemma:beta_property_test,thm:conf_mod_eq}). For this purpose, the set of \emph{compatible} equivalence relations (c.f. \Cref{def:congruence_relation,def:split_property,def:compatibility}) has been identified to behave well with this confluence criterion for CHR. 
When an equivalence relation has been shown to be compatible and maintains the invariant, it can be used directly for any program. In practice, it seems to be desirable that the equivalence relation maintains the invariant.

The approach is directly applicable for a non-trivial example (c.f. \Cref{ex:set_items_confluence}). It has been tested for other examples which indicates that the defined subset of equivalence relations is actually meaningful. Decidability of the $\alpha$ property is maintained. For some invariants, the set of minimal extensions can be infinitely large and decidability is lost.  Although the $\beta$ property leads to an infinite number of states that have to be considered in general, the proofs are simplified tremendously, as only states equivalent to the finite number of rule states have to be considered. 

In many cases it may suffice to only use an equivalence relation without an invariant. The problems originating from invariants are inexistent for those cases and our approach yields a sufficient and necessary criterion for confluence modulo equivalence without invariants. 

For the future, we want to investigate how our approach can be unified with the meta-level approach \cite{christiansen_2017} or other proof methods in the context of confluence such as case splitting. Furthermore, it could be interesting how non-confluent programs can be completed such that they become confluent modulo equivalence.

\section*{Acknowledgements}

 The authors would like to thank Henning Christiansen and Maja H. Kirkeby for the valuable discussions and ideas for future work.

%% file: confmodeq-appendix.tex
\section{Proofs}

\subsection{Merge Compatibility of $\seteq$}

\label{app:seteq_compatibility_proof}

The multi-set equivalence relation $\seteq$ from \Cref{ex:set_items_seteq} is $\merge$-compatible.

\begin{proof}
\begin{description}
 \item[Congruence Relation] Let $[\sigma], [\sigma'], [\rho], [\rho']$ be CHR states. From the definition it follows that if $[\sigma] \seteq [\sigma']$ , the number of constraints (and in particular \emph{mset} constraints) is equivalent. This means that if there is no \emph{mset} constraint in $[\sigma]$ or $[\rho]$, then there is none in $[\sigma']$ or $[\rho']$. We use induction over the number $n$ of constraints in the constraint store of $[\sigma]$. Since $\merge_\mathbb{V}$ is commutative (c.f. \cite[p. 51 sqq.]{raiser_phdthesis10}), this can be done w.l.o.g.
 
  \begin{description}
 \item[Base Case ($n = 0$)]
 
 Let  $[\sigma] = [\langle \emptyset ; \mathbb{B}_\sigma ; \mathbb{V}_\sigma \rangle]$. Since $[\sigma] \seteq [\sigma']$, we get $[\sigma'] = [\langle \emptyset ; \mathbb{B}_\sigma ; \mathbb{V}_\sigma \rangle]$. Let $[\rho] = [\langle \mathbb{G}_\rho ; \mathbb{B}_\rho ; \mathbb{V}_\rho \rangle] \seteq [\rho'] = [\langle \mathbb{G}_{\rho'} ; \mathbb{B}_{\rho'} ; \mathbb{V}_{\rho} \rangle]$. Then for all $\mathbb{V}:$ $[\sigma] \merge_\mathbb{V} [\rho] = [\langle \mathbb{G}_\rho ; \mathbb{B}_\sigma \land \mathbb{B}_\rho ; (\mathbb{V}_\sigma \cup \mathbb{V}_\rho) \setminus \mathbb{V} \rangle] \seteq [\langle \mathbb{G}_{\rho'} ; \mathbb{B}_\sigma \land \mathbb{B}_{\rho'} ; (\mathbb{V}_\sigma \cup \mathbb{V}_{\rho}) \setminus \mathbb{V} \rangle] = [\sigma'] \merge_\mathbb{V} [\rho']$. 

 \item[Induction Step ($n \rightarrow n + 1$)]
 
 We add a constraint $c$ to the constraint store of a state $[\delta]$ with $n$ constraints. If $c$ is not a \emph{mset} constraint, it is clear that the proposition holds. Let $c = \constr{mset}(L)$ and $[\sigma] = [\langle \{ c \} \uplus \mathbb{G}_\sigma ; \mathbb{B}_\sigma ; \mathbb{V}_\sigma \rangle]$. Since $[\sigma] \seteq [\sigma']$, we get by definition that 
 \begin{equation*}
 [\sigma'] = [\langle \{ c' \} \uplus \mathbb{G}_{\sigma'} ; \mathbb{B}_{\sigma} ; \mathbb{V}_{\sigma} \rangle]
 \end{equation*}
 for a constraint $c' = \constr{mset}(L')$   and $L$ is a permutation of $L'$.  Let 
 \begin{align*}
  [\delta] := &  [\langle \mathbb{G}_{\sigma} ; \mathbb{B}_{\sigma} ; \mathbb{V}_{\sigma} \rangle] \\
     \seteq   &  [\langle \mathbb{G}_{\sigma'} ; \mathbb{B}_{\sigma} ; \mathbb{V}_{\sigma} \rangle] =: [\delta']
 \end{align*}
 and
  \begin{align*}
  [\rho] := &  [\langle \mathbb{G}_{\rho} ; \mathbb{B}_{\rho} ; \mathbb{V}_{\rho} \rangle] \\
     \seteq   &  [\langle \mathbb{G}_{\rho'} ; \mathbb{B}_{\rho'} ; \mathbb{V}_{\rho'} \rangle] =: [\rho'].
 \end{align*}

 By using the induction hypothesis, it follows that  $[\delta] \merge_\mathbb{V} [\rho] \seteq [\delta'] \merge_\mathbb{V} [\rho']$ and hence 
 \begin{align*}
   &[\sigma] \merge_\mathbb{V} [\rho] \\
=  & [\langle \{c \} \uplus \mathbb{G}_{\sigma} \uplus \mathbb{G}_{\rho}; \mathbb{B}_{\sigma} \land  \mathbb{B}_{\rho}; (\mathbb{V}_{\sigma} \cup \mathbb{V}_{\rho}) \setminus \mathbb{V} \rangle]\\
\seteq & [\langle \{ c' \} \uplus \mathbb{G}_{\sigma'} \uplus \mathbb{G}_{\rho'}; \mathbb{B}_{\sigma} \land  \mathbb{B}_{\rho'}; (\mathbb{V}_{\sigma} \cup \mathbb{V}_{\rho'}) \setminus \mathbb{V} \rangle] \\
=  & [\sigma'] \merge_\mathbb{V} [\rho']
 \end{align*}
 by definition of $\merge_\mathbb{V}$ and $\seteq$.
 \end{description}
  
 \item[Split Property]
 
 The property is proven by induction.
 \begin{description}
 \item[Base Case ($n = 0$)]
 
 Let $[\sigma] = [\langle \emptyset ; \mathbb{B} ; \mathbb{V} \rangle] =  [\sigma_1] \merge_\mathbb{V} [\sigma_2]$ and $[\sigma] \seteq [\rho]$.
 If $[\sigma]$ has an empty goal store, it does not contain any $\constr{mset}$ constraints and hence $[\sigma] = [\rho]$. There is a trivial split $[\rho] = [\sigma_1] \merge_\mathbb{V} [\sigma_2]$.
 
 \item[Induction Step ($n \rightarrow n + 1$)]
 
 Let $[\sigma] = [\langle \{c\} \uplus \mathbb{G}_{\sigma} ; \mathbb{B}_{\sigma} ; \mathbb{V}_{\sigma} \rangle] =  [\sigma_1] \merge_\mathbb{V} [\sigma_2]$ for a constraint $c$ and $[\sigma] \seteq [\rho]$. There are two cases:
 \begin{enumerate}
  \item $c$ is not a $\constr{mset}$ constraint. Therefore, $[\sigma] = [\rho]$. There is a trivial split $[\rho] = [\sigma_1] \merge_\mathbb{V} [\sigma_2]$.
  \item $c = \constr{mset}(L)$ for some $L$. Then, since $[\sigma] \seteq [\rho]$, there is a $c' = \constr{mset}(L')$ such that $[\rho] = [\langle \{c'\} \uplus \mathbb{G}_{\rho} ; \mathbb{B}_{\rho} ; \mathbb{V}_{\rho} \rangle]$ where $L$ is a permutation of $L'$ and 
  \begin{align*}
    [\delta_\sigma] :=  & [\langle \mathbb{G}_{\sigma} ; \mathbb{B}_{\sigma} ; \mathbb{V}_{\sigma} \rangle] \\
                 \seteq & [\langle \mathbb{G}_{\rho} ; \mathbb{B}_{\rho} ; \mathbb{V}_{\rho} \rangle] =: [\delta_\rho].
  \end{align*}

 From the induction hypothesis it follows that if $[\delta_\sigma] = [\delta_\sigma^1] \merge_\mathbb{V} [\delta_\sigma^2]$, then there are $[\delta_\rho^1]$ and $[\delta_\rho^2]$, such that $[\delta_\rho] = [\delta_\rho^1] \merge_\mathbb{V} [\delta_\rho^2]$. The constraint $c$ from $[\sigma]$ can now be part of either $[\sigma_1]$ or $[\sigma_2]$. Let w.l.o.g. $c$ be part of $[\sigma_1]$, i.e. 
 \begin{equation*}
  [\sigma_1] = [\langle \{ c \} \uplus \mathbb{G}_{\delta_\sigma^1} ; \mathbb{B}_{\delta_\sigma^1} ; \mathbb{V}_{\delta_\sigma^1} \rangle]
 \end{equation*}
 where
 \begin{equation*}
   [\delta_\sigma^1] = [\langle \mathbb{G}_{\delta_\sigma^1} ; \mathbb{B}_{\delta_\sigma^1} ; \mathbb{V}_{\delta_\sigma^1} \rangle].
 \end{equation*}
 In this case, $[\sigma_2] = [\delta_\sigma^2]$.
 
 Analogously, we have that
 \begin{equation*}
  [\rho_1] = [\langle \{ c \} \uplus \mathbb{G}_{\delta_\rho^1} ; \mathbb{B}_{\delta_\rho^1} ; \mathbb{V}_{\delta_\rho^1} \rangle]
 \end{equation*}
 where
 \begin{equation*}
   [\delta_\rho^1] = [\langle \mathbb{G}_{\delta_\rho^1} ; \mathbb{B}_{\delta_\rho^1} ; \mathbb{V}_{\delta_\rho^1} \rangle].
 \end{equation*}

 Since $[\delta_\rho^1] \seteq [\delta_\sigma^1]$ by induction hypothesis and $L$ is a permutation of $L'$, by definition of $\seteq$ in \Cref{ex:set_items_seteq} it is clear that $[\sigma_i] \seteq [\rho_i]$ for $i = 1,2$.
\qed
 \end{enumerate}
 \end{description}
\end{description}
\end{proof}

\subsection{Proof for $\alpha$ Property Test}
\label{app:alpha_property_test_proof}

In this section, the $\alpha$ property test from \Cref{lemma:alpha_property_test} is proven.

Let $\mathcal{P}$ be a CHR program, $\inv$ an invariant, $\eq$ a congruence relation and let $\mathcal{M}^\inv([\sigma])$ be well-defined for all overlaps $\sigma$, then:
$\mathcal{P}$ has the $\alpha$ property with respect to $\inv$ and $\eq$ if and only if for all overlaps $\sigma$ with critical pairs $(\sigma_1,\sigma_2)$ and all $[\sigma_\mathrm{m}] \in \mathcal{M}^\inv([\sigma])$ holds $([\sigma_1] \merge [\sigma_\mathrm{m}] \joins^\eq [\sigma_2] \merge [\sigma_\mathrm{m}])$.

\begin{proof}
The $\alpha$ property test coincides with the invariant-based confluence test first presented for CHR in \cite[p. 86, lemma 14.11]{raiser_phdthesis10}. However, the proof has to be adapted in the last step as joinability now allows states to join modulo an equivalence relation. 

``$\Rightarrow$'': 
This follows directly from \cref{def:alpha_beta_property} and \cref{lemma:joinability_monotonicity}. 

``$\Leftarrow$'': 
Let $[\sigma], [\sigma_1]$ and $[\sigma_2]$ be CHR states where $\inv([\sigma])$ holds and $[\sigma] \mapsto_{r_1} [\sigma_1]$ for some rule $r_1$ and $[\sigma] \mapsto_{r_2} [\sigma_2]$ for some rule $r_2$. By \cref{def:overlap}, there exists an overlap state $\sigma_\mathrm{o} = \langle \_ ; \_ ; \mathbb{V} \rangle$ of rule $r_1$ and $r_2$ where $\mathbb{V}$ contains all variables from $r_1$ and $r_2$ such that for some $[\delta] := [ \langle \mathbb{G} ; \mathbb{B} ; \mathbb{V'} \rangle] \text{ it holds that } [\sigma] = [\sigma_\mathrm{o}] \merge_\mathbb{V} [\delta].$
The variables $\mathbb{V}$ from the rules $r_1$ and $r_2$ are not part of $[\sigma]$ and are therefore removed by the merging $\merge_\mathbb{V}$. Due to monotonicity (c.f. \cref{lemma:monotonicity}), we have that
\begin{itemize}
 \item $[\sigma_\mathrm{o}] \mapsto_{r_1} [\sigma_1']$ with $[\sigma_1] = [\sigma_1'] \merge_\mathbb{V} [\delta],$ and
 \item $[\sigma_\mathrm{o}] \mapsto_{r_2} [\sigma_2']$ with $[\sigma_2] = [\sigma_2'] \merge_\mathbb{V} [\delta]$.
\end{itemize}
If no such overlap exists, the two rule applications are independent and therefore trivially joinable.

As $\inv([\sigma])$ holds, we have that $[\delta] \in \Sigma^\inv([\sigma_\mathrm{o}])$ and therefore there is a element in the set of minimal extensions that is less than or equal to $[\delta]$, i.e. $\exists [\sigma_\mathrm{m}] \in \mathcal{M}^\inv([\sigma_\mathrm{o}]) . [\sigma_\mathrm{m}] \ltstate [\delta]$. This means that there is a minimal element $[\sigma_\mathrm{m}]$ in the set of extensions of the overlap state $[\sigma_\mathrm{o}]$ that extend $[\sigma_\mathrm{o}]$ such that the invariant holds.

It follows by definition of $\ltstate$ that $\exists [\delta'] . [\delta] = [\sigma_\mathrm{m}] \merge [\delta']$ and hence $[\sigma] = [\sigma_\mathrm{o}] \merge_\mathbb{V} ([\sigma_\mathrm{m}] \merge [\delta'])$. By \cref{lemma:special_assoc_merge}, we get $[\sigma] = ([\sigma_\mathrm{o}] \merge [\sigma_\mathrm{m}]) \merge_\mathbb{V} [\delta']$. Analogously, since $\inv([\sigma_i]), i = 1, 2$ holds due to the definition of an invariant, we find that $[\sigma_i] = ([\sigma_i'] \merge [\sigma_\mathrm{m}]) \merge_\mathbb{V} [\delta']$ for $i = 1, 2$. 

At this point the proof differs from confluence without an equivalence relation. By the precondition we now only have that $([\sigma_1'] \merge [\sigma_\mathrm{m}]) \joins^\eq ([\sigma_2'] \merge [\sigma_\mathrm{m}])$, i.e. modulo an equivalence relation. Since $[\sigma_1] = ([\sigma_1'] \merge [\sigma_\mathrm{m}]) \merge_\mathbb{V} [\delta']$ and $[\sigma_2] = ([\sigma_2'] \merge [\sigma_\mathrm{m}]) \merge_\mathbb{V} [\delta']$, we can apply \cref{lemma:joinability_monotonicity} due to reflexivity of $\eq$ (i.e. $[\delta'] \eq [\delta']$) and get $([\sigma_1] \joins^\eq [\sigma_2])$.
\end{proof}